\documentclass[12pt]{article}

\usepackage{sbc-template}

\usepackage[utf8]{inputenc}

\usepackage{fancyhdr}
\usepackage{fancyref}
\usepackage[cmex10]{amsmath}
\usepackage{color}
\usepackage[colorlinks,bookmarksopen,bookmarksnumbered,citecolor=blue,urlcolor=red]{hyperref}
\usepackage{algorithm}
\usepackage{algorithmicx}
\usepackage{algpseudocode}
\usepackage{amssymb}
\usepackage{amsbsy}
\usepackage{ulem}
\usepackage{upgreek}
\usepackage{graphicx}
\usepackage{subfigure}
\usepackage{epsfig}
\usepackage{epstopdf}
\usepackage{comment}

\graphicspath{{figures/}}

\title{R-BBG$_2$: Recursive Bipartition of Bi-connected Graphs}

\author{Ngoc-Tu Nguyen\inst{1} and Zhi-Li Zhang\inst{1} }

\address{Department of Computer Science \& Engineering\\
University of Minnesota, Twin Cities\\
Minneapolis MN 55455, USA.
\email{nguy3503@umn.edu and zhzhang@cs.umn.edu }
}

\begin{document}

\newtheorem{thm}{Theorem}
\newtheorem{lma}{Lemma}
\newtheorem{defi}{Definition}
\newtheorem{remark}{Remark}
\newtheorem{proof}{Proof}
\newtheorem{co}{Corollary}
\newtheorem{po}{Property}
\newtheorem{cl}{Claim}

\maketitle

\begin{abstract}
Given an undirected graph $G(V, E)$, it is well known that partitioning a graph $G$ into $q$ {\it connected} subgraphs of  equal or specificed sizes is in general NP-hard problem. On the other hand, it has been shown that the q-partition problem is solvable in polynomial time for q-connected graphs. For example, efficient polynomial time algorithms for finding 2-partition (bipartition) or 3-partition of 2-connected or 3-connected have been developed in the literature. In this paper, we are interested in the following problem: given a bi-connected graph $G$ of size $n$, can we partition it into two (connected) sub-graphs, $G[V_1]$ and $G[V_2]$ of sizes $n_1$ and $n_2$ such as both $G[V_1]$ and $G[V_2]$ are also bi-connected (and $n_1+n_2=n$)? We refer to this problem as  the recursive bipartition problem of bi-connected graphs, denoted by {\it R-BBG$_2$}. We show that a ploynomial algorithm exists to both decide the recursive bipartion problem {\it R-BBG$_2$} and find the corresponding bi-connected subgraphs when such a recursive bipartition exists.
\end{abstract}

\begin{keywords}
Bipartition, bi-connected, recursive bipartition.
\end{keywords}

\section{Introduction} \label{section:introduction}
Nowadays, graph partitioning has been widely studied and applied for image processing, data bases, operating systems and cluster analysis applications \cite{LUCERTINI1993227,MARAVALLE1997217,BECKER1983101,Lucertini1989}. Graph partitioning is a class of the optimization problem, where we expect to break a given graph into parts satisfying
certain constraints. Given a graph $G(V,E)$ with weights on the vertices, the concept of {\it partition} is formalized as follows: find a partition of the vertex set $V$ into $q$ parts $V_1, V_2, \ldots, V_q$ such that the weight of every part is as equal as possible. The problem is called {\it Max Balanced Connected q-Partition Problem ($BCP_q$)} \cite{Wang2013,Becker2001,Perl:1981:MTP:322234.322236}.

The unweighted version of $BCP_q$ is a special case in which all vertices have weight $1$ and the graph is expected to be bipartitioned into two parts $V_1$ and $V_2$ in which each induces a connected subgraph of G, namely 1-BCP$_2$.
For simplicity, throughout this paper we use notation $BCP_2$ to refer to bipartition problem of unweighted version of $BCP_2$. The $BCP_2$ is NP-hard in general \cite{DYER1985139} and will be defined formally in Section \ref{section:preliminaries}.
However, $BCP_2$ is solvable in polynomial time for bi-connected graphs \cite{SUZUKI1990227}. In \cite{SUZUKI1990227}, the authors introduced a linear algorithm for bipartition of bi-connected graphs for $BCP_2$. Later, the polynomial algorithms for $q \geq 2$ were also obtained. In $1990$, the authors in \cite{Miyano} proposed an $O(n^2)$ algorithm to find 3-partition for 3-connected graph. In \cite{NAKANO1997315}, Nakano, Rahman and Nishizeki introduced a linear-time algorithm for four-partitioning four-connected planar graphs. The authors in \cite{Ma1994} covered the problem in general. They presented an $O(q^2n^2)$ algorithm to find q-partition of q-connected graph.

Especially, many researchers pay a lot of attention to solve the bipartition problem. Most of them rely on geometrical properties of some specific graphs \cite{10.1007/978-1-4613-8369-7_3, Hendrickson:1995:ISG:203046.203060, doi:10.1137/S1064827593255135} such as mesh and are thus only applicable for some limited problems. Although there are polynomial algorithms of finding a 2-partition (bipartition), 3-partition, or q-partition for the undirected bi-connected, 3-connected, or q-connected graph, respectively, there is no general algorithm of recursive finding partition of graph. The requirement that makes the problem interesting, and also difficult, is that each part has to induce a bi-connected subgraph of $G$.
In this paper a decision algorithm of recursive finding bipartition of bi-connected graph is investigated.

The remaining of this paper is organized as follows. In Section \ref{section:preliminaries} we give some definitions and establish
the notation. In Section \ref{section:bipartition} we show some properties of bipartition with bi-connected graphs. We also present a decision algorithm to solve the recursive bipartition problem of bi-connected graphs, namely decision algorithm in Section \ref{section:decision}. Section \ref{section:conclude} concludes the paper.

\section{Definitions and Notation} \label{section:preliminaries}

Let $G(V, E)$ be an undirected bi-connected graph in which for every vertex
$u \in V$, $G - \{u\}$ is connected. Let $|V| = n$ ($n \geq 6$) be the number of vertices/nodes. In this paper we shall use terminology vertex and node, bi-connected and 2-connected interchangeably. For any subset $V' \subseteq V$ , we denote by $G[V']$ the subgraph of $G$ induced by $V'$; and we denote by $w(V')$ the number of vertices in $V'$, that is, $w(V')=|V'|$. For any $v \in V$, let $N_v$ be a set of $v$'s neighbors in which any $v' \in N_v$ is adjacent to $v$.
We first show the decision version of bipartition problem as follows:

\textbf{INSTANCE:} $G(V, E)$, $s_1, s_2 \in V$, $(s_1, s_2) \in E$, $n_1, n_2 \geq 0$.

\textbf{QUESTION:} Does there exist a bipartition $V_1, V_2$ of vertex set $V$ such that $s_1 \in V_1$, $s_2 \in V_2$, $w(V_1) = n_1$ and $w(V_2) = n_2$, and each of $V_1$ and $V_2$ induces a connected subgraph $G[V_1]$ and $G[V_2]$ of $G$?

\begin{figure}[h]
\center \subfigure{\includegraphics[width=8cm]{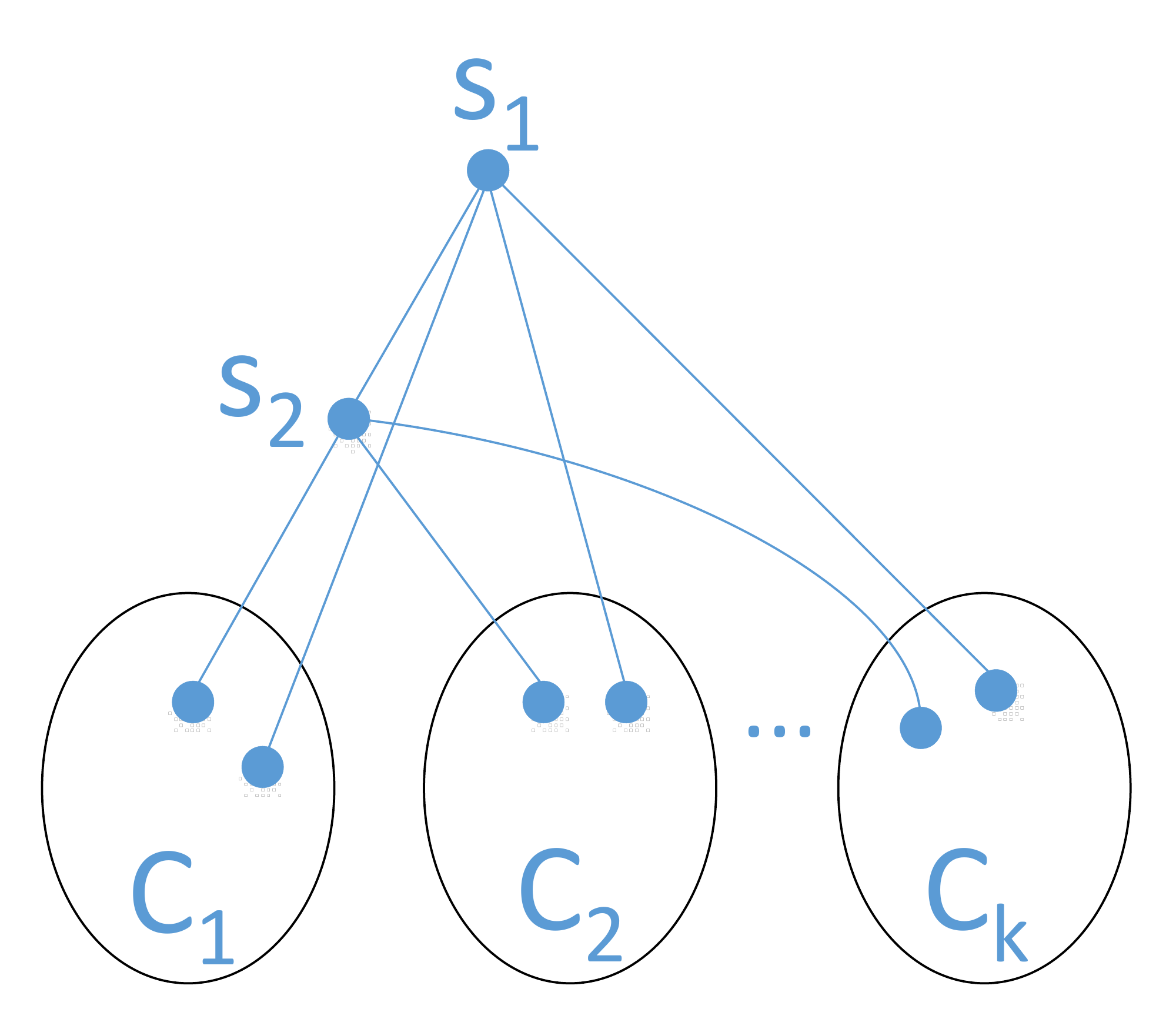}}
\caption{An example of $k$ $\textit{components}$ of $G$.} \label{Fig:partition}
\end{figure}

In this research we expect that a bipartition of graph $G$ is a division of its vertex set into two subsets $V_1$ and $V_2$ of exactly equal sizes, so throughout this paper we assume that $n_1 = n_2 = \frac{n}{2}$ (if graph has even number of vertices), or $n_1 = \lceil\frac{n}{2}\rceil$ and $n_2 = \lceil\frac{n}{2}\rceil - 2$ (if graph has odd number of vertices).
The bipartition problem is NP-hard in general \cite{DYER1985139}. Although it is NP-hard to tell whether there is a solution satisfying above conditions; however, in this case, as indicated by Theorem \ref{tr:lovasz} if the graph G is bi-connected there is always such a solution and it can be found easily in linear
time using the st-numbering between two nodes.

\begin{thm}\textbf{(Lovasz \cite{Lovasz1977241})} \label{tr:lovasz}
Let $G$ be a q-connected graph with $n$ vertices, $q \geq 2$, and let $n_1, n_2,\ldots, n_q$ be positive natural numbers such that $n_1 + n_2 + \ldots + n_q = n$. Then $G$ has a connected q-partition $(V_1, V_2, \ldots , V_q)$ such that $|V_i| = n_i$ for $i = 1, 2,\ldots, q$.
\end{thm}

Even though there are polynomial algorithms of finding a 2-partition (bipartition) \cite{SUZUKI1990227}, 3-partition \cite{Miyano}, or q-partition \cite{Ma1994} for the undirected bi-connected, 3-connected, or q-connected graph, respectively, there is no general algorithm of recursive finding partition of graph. The requirement that makes the problem interesting, and also difficult, is that each part has to induce a bi-connected subgraph of $G$.
In this paper a decision algorithm of recursive finding bipartition of bi-connected graph is investigated.

\textbf{RESEARCH PROBLEM:} In this paper we will concentrate on bi-connected graph and show that we can {\it recursively} construct good bipartitions in which subgraphs are bi-connected and most equal size to each other in every bipartition, namely {\it R-BBG$_2$} problem (No related work guarantees that each part of bipartition has to induce a bi-connected subgraph).

We realize that the solution of {\it R-BBG$_2$} problem can be found by recursively applying bipartition as Algorithm \ref{alg:recursive}, the problem becomes how to optimally bipartition a bi-connected graph into two bi-connected and equal sized pieces in the first step.

\begin{algorithm}(bi-connected graph $G(V,E)$, $s_1, s_2 \in V$, $s_1 \neq s_2$, $n_1, n_2$)
\caption{Recursive algorithm.} \label{alg:recursive}
\begin{algorithmic}[1]

    \State Find: an "optimal" bipartition $G_1[V_1], G_1[V_2]$ of $G$

    \If {$w(V_1) \geq 3$ or $w(V_2) \geq 3$}

        \begin{itemize}
      		\item Recursively bipartition $G_1[V_1]$
      		\item Recursively bipartition $G_1[V_2]$
    	\end{itemize}
    	
    \EndIf

    \State Return the subgraphs: $G_1[V_1], G_1[V_2],\ldots,G_q[V_1], G_q[V_1]$;
\end{algorithmic}
\end{algorithm}

For simplicity of conveying the idea, note that the following notations and results are only for bi-connected graph. Suppose $C_1,\ldots,C_k$ are $k$ components of $G$ induced by removals of $s_1$, $s_2$, and all relative links of $s_1$ and $s_2$.
Throughout this paper, whenever we mention the removals of $s_1$ and $s_2$, it means we refer to the removals of $s_1$, $s_2$ and all relative links of $s_1$ and $s_2$.
Consider for example the graph $G$ shown in Fig. \ref{Fig:partition}. This example graph is divided into $k$ components (each is shown in a circle) by the removals of $s_1$ and $s_2$.

\section{Bipartition and bi-connected graph} \label{section:bipartition}

In general, a vertex $u$ of $G$ is an articulation point (or cut vertex) if $u'$s removal
disconnects graph. It is obvious that $G$ is bi-connected if it is connected and has no articulation points.
By the fact that $G$ is bi-connected, we have following property:

\begin{po} \label{po:2-connected}
G is a bi-connected graph if and only if for any two vertices $u,v \in V$, there are $2$ paths in $G$ joining $u$ and $v$ without intersecting inner vertex.
\end{po}

It also obvious that $C_i \bigcap N_{s_1} \geq 1$ and $C_i \bigcup N_{s_2} \geq 1$ for $1 \leq i \leq k$. Otherwise, either $s_1$ or $s_2$ is an articulation point of $G$.

\begin{figure}[h]
\center
\subfigure[]{\includegraphics[width=7.5cm]{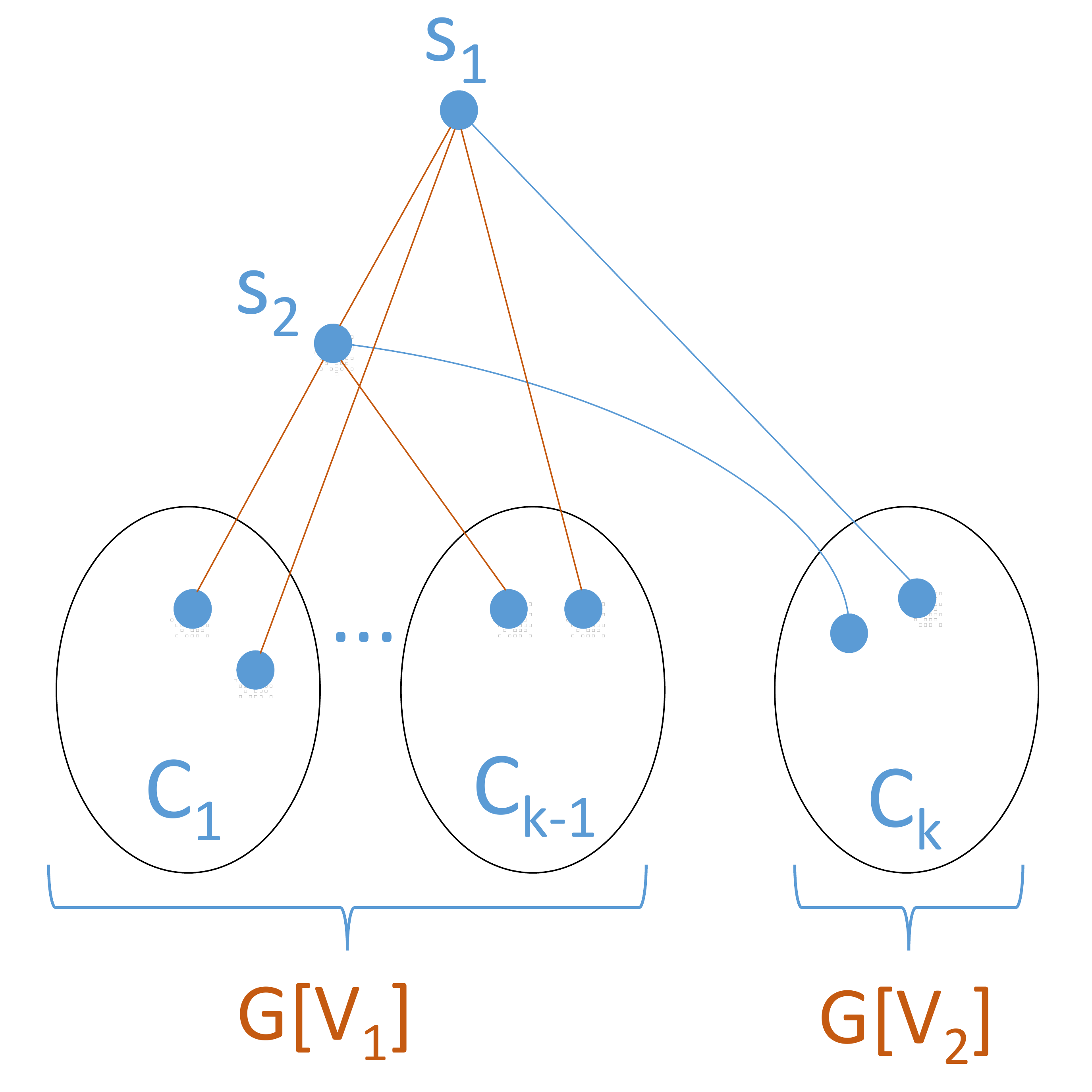} \label{Fig:s1-s1-same-1}}
\subfigure[]{\includegraphics[width=7.5cm]{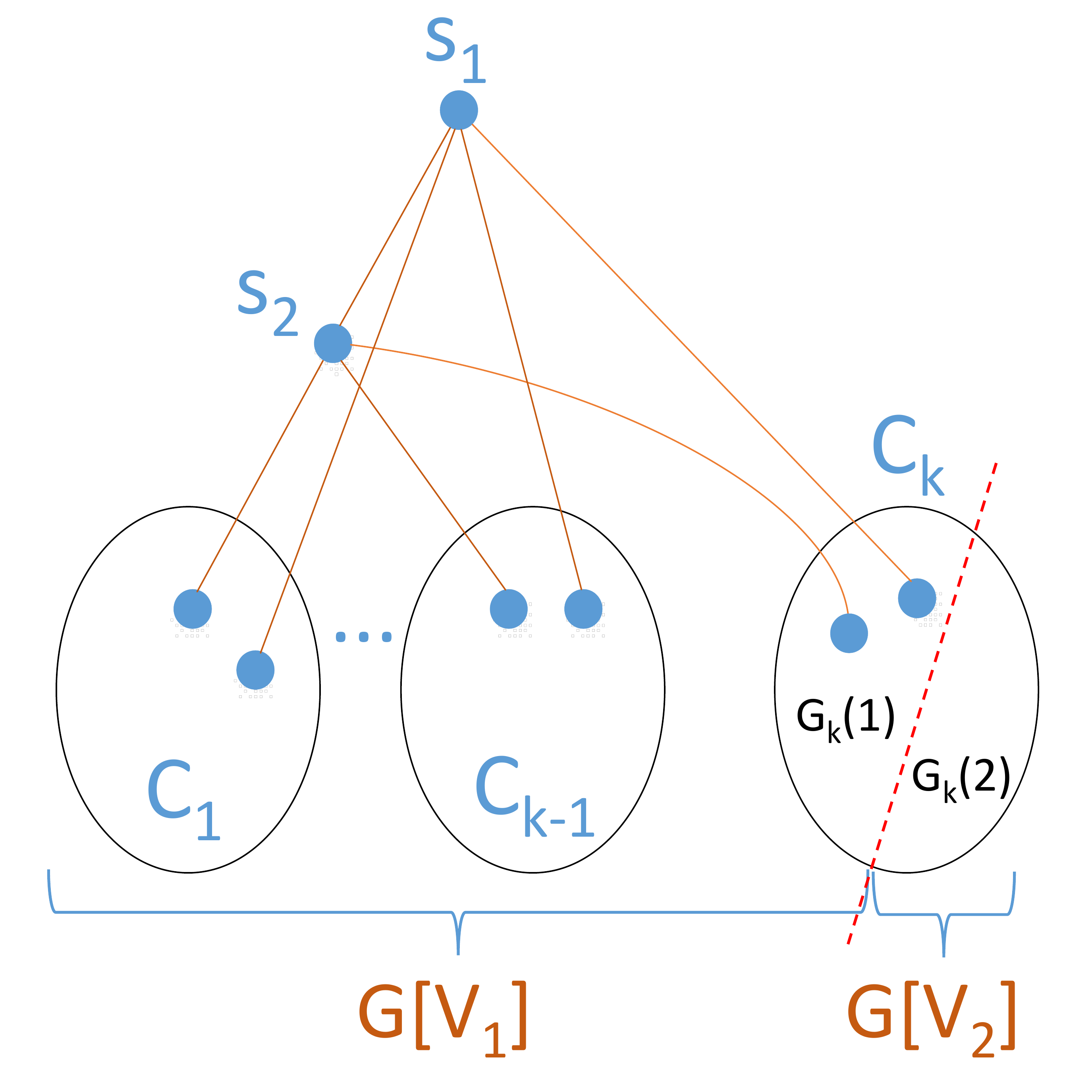} \label{Fig:s1-s1-same-2}}
\caption{An example to illustrate for Lemma \ref{lm:k-32}.}
\end{figure}

\begin{lma} \label{lm:k-3}
If $k > 2$, then $G$ cannot be partitioned into $2$ bi-connected subgraphs (each contains at least 3 nodes) such that $s_1 \in V_1$ and $s_2 \in V_2$.
\end{lma}

\begin{proof}
Let $G[V_1]$ and $G[V_2]$ be any bipartition of $G(V,E)$, where $s_1 \in V_1$, $s_2 \in V_2$, and $w(V_i) \geq 3$ for $i= 1,2$. Without loss of generality, suppose:

For $i \neq j$ $C_i \bigcup C_j \bigcup \{s_1\}$ induces $G[V_1]$.

It is easy to see $s_1$ is an articulation point of $G[V_1]$ and also an articulation point of $G$. This constitutes a contradiction because $G(V,E)$ is bi-connected and there is no articulation point exists in $G$, and thus completes the proof.
\end{proof}

\begin{figure}
\center
\subfigure[]{\includegraphics[width=7.5cm]{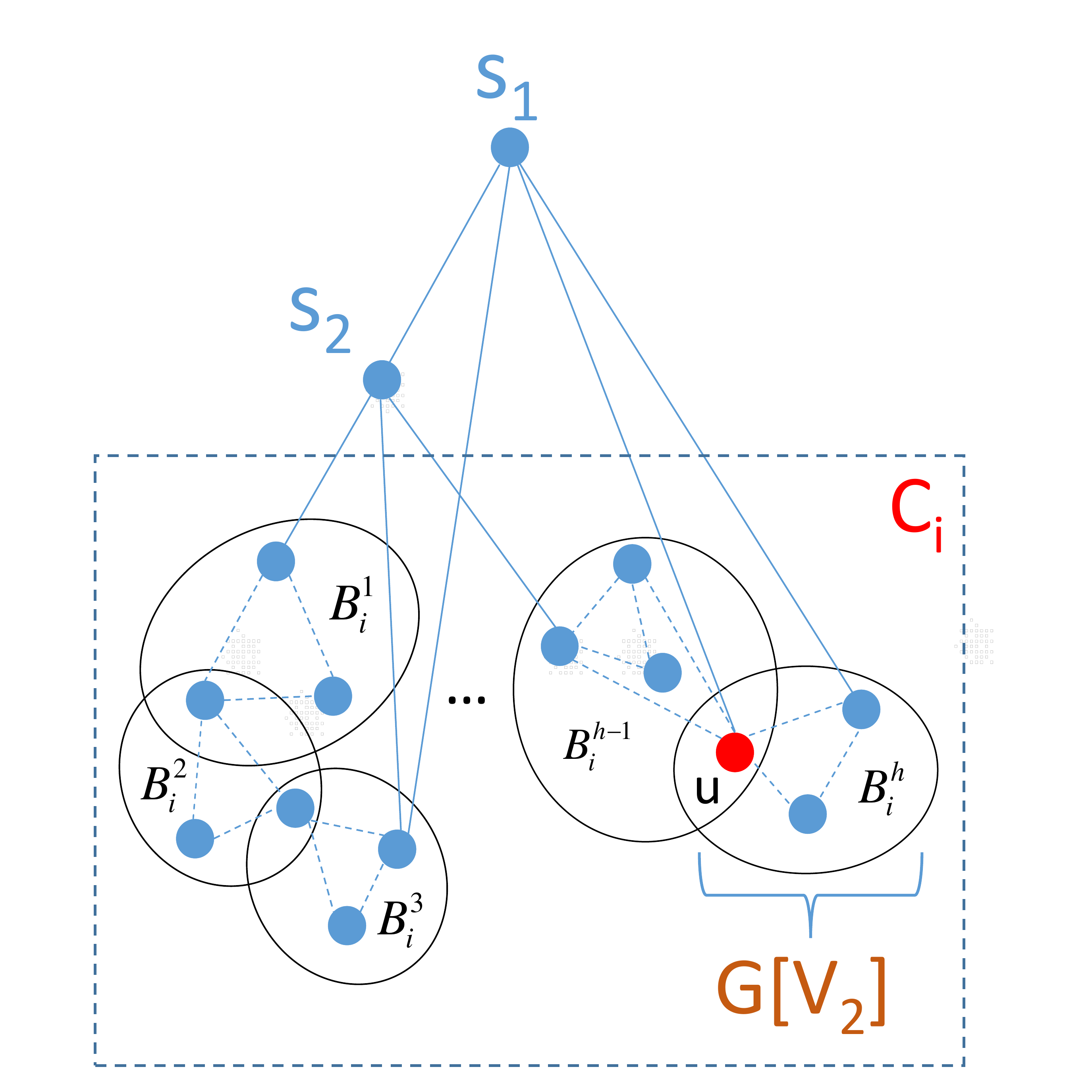} \label{Fig:leaf-block-cut}}
\subfigure[]{\includegraphics[width=7.5cm]{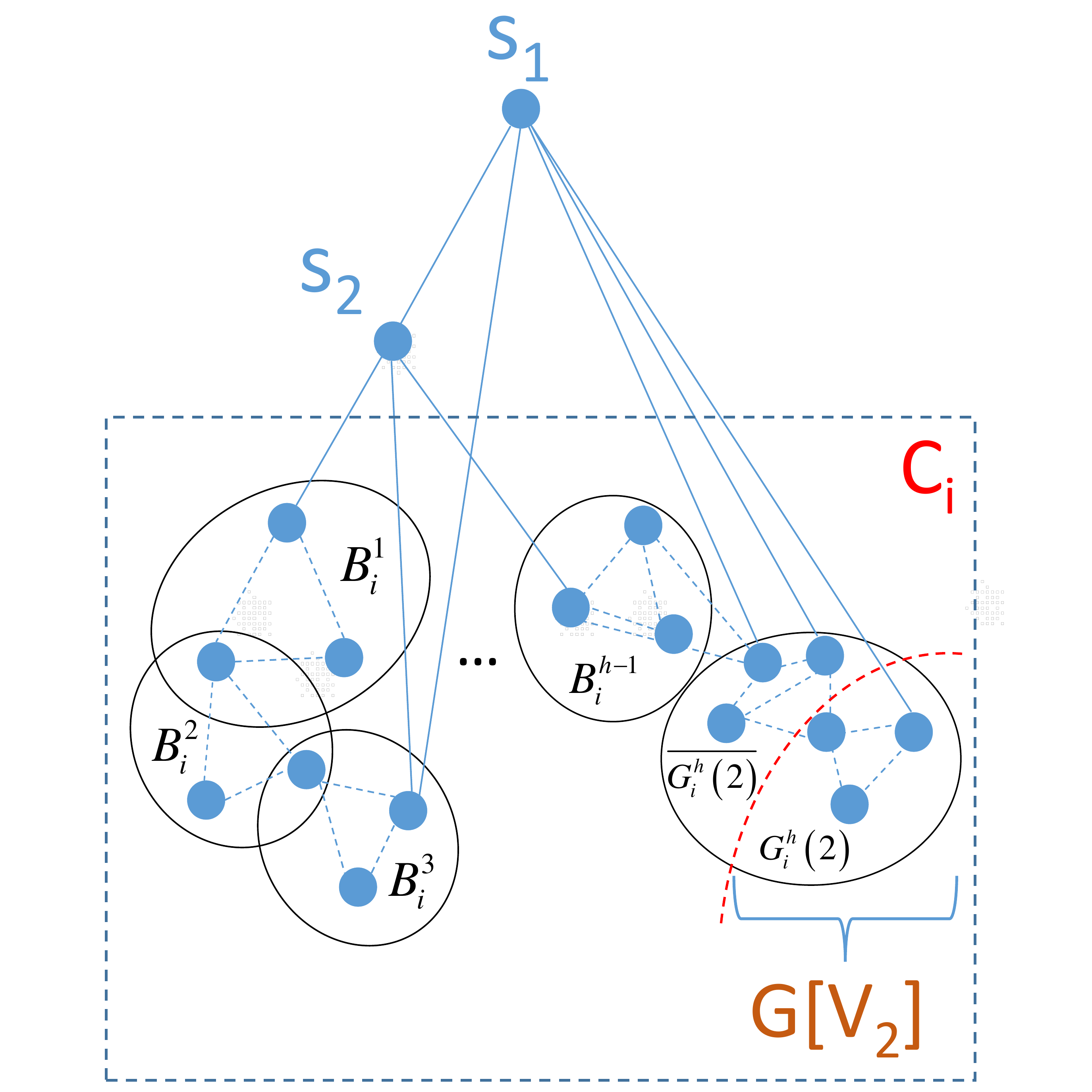} \label{Fig:leaf-block-cut-2}}
\caption{An example of {\it blocks} of component of graph.}
\end{figure}

\begin{lma} \label{lm:k-32}
If $k > 2$, suppose $G$ can be partitioned into two bi-connected subgraphs $G[V_1]$ and $G[V_2]$.
Note that Lemma \ref{lm:k-3} indicates that either $s_1, s_2 \in V_1$ or $s_1, s_2 \in V_2$. Without loss of generality, assume that $s_1, s_2 \in V_1$. Then $G[V_1]$ contains $k-1$ out of $k$ components $C_1,C_2, \ldots, C_k$. We also assume that $w(C_i) \leq n_1$ $\forall$ $1 \leq i \leq k-1$. Furthermore, we have either:

\begin{enumerate}
  \item ${{V}_{1}}=\bigcup\limits_{1\le i\le k-1}{{{C}_{i}}}\bigcup \left\{ {{s}_{1}},{{s}_{2}} \right\}$, $G[V_1]$ is bi-connected, and $G[V_2] = C_k$ is bi-connected (see Fig. \ref{Fig:s1-s1-same-1} for an illustration). $\bf{(I)}$\\
  \item or $V_k = C_k \bigcup \{s_1,s_2\}$, $G[V_k]$ then can be partitioned into two bi-connected subgraphs $G_k(1)$ and $G_k(2)$ in which $G_k(1) \bigcup G_k(2) = G[V_k]$, $G_k(2) = G[V_2]$, and $s_1,s_2 \in G_k(1)$ (see Fig. \ref{Fig:s1-s1-same-2} for an illustration). $\bf{(II)}$
\end{enumerate}
\end{lma}

\begin{proof}
Note that since $s_1, s_2 \in G[V_1]$, we have that if there exists $i$ and $j$ such that $G[V_2] \bigcap C_i \neq \emptyset$, $G[V_2] \bigcap C_j \neq \emptyset$; on the other words, $G[V_2]$ includes two or more components of $G$, then $G[V_2]$ is disconnected. Hence, there exists exactly one components $C_i$ such that $G[V_2] \bigcap C_i \neq \emptyset$ (in this case, let $i = k$, we have $G[V_2] = C_k$). The conditions in $\bf{(I)}$ and $\bf{(II)}$ are required for $G[V_1]$ and $G[V_2]$ to be bi-connected.
\end{proof}

In this paper we are frequently dealing with {\it block}, and in the following, we will further characterizes $\bf{(I)}$ and $\bf{(II)}$ in terms of {\it blocks} of every component $C_i$, so let us define {\it block} and {\it block tree} of graph.
A block of $G$ is defined as a maximal bi-connected subgraph of $G$.
The block tree of $G$ \cite{cut1, cut2} is a bipartite graph $B(G)$ with bipartition
$(\mathcal{B}, \mathcal{A})$, in which $\mathcal{B}$ is the set of blocks of $G$, $\mathcal{A}$ is the set of articulation points of G. It is easy to see that a block $B \in \mathcal{B}$ and an articulation point $u \in \mathcal{A}$ are adjacent in $B(G)$ if and only if $B$ contains $u$. A block $B$ of $G$ is a leaf block if $B$ is a leaf of the block tree $B(G)$. Note that a
leaf block contains at most one articulation point of $G$.

\begin{lma} \label{lm:leaf-block}
For $1 \leq i \leq k$, suppose $C_i$ can be decomposed into $h$ blocks.
Let $\mathcal{B}_i = \{B_i^1, B_i^2, \ldots, B_i^h\}$ be the set of blocks of $C_i$.
Let $\mathcal{A}_i$ also be the set of articulation points of $C_i$.
Without loss of generality, assume that $B_i^h$ be a leaf block. Then, we have:

$|B_i^h \bigcap N_{s_1}| + |B_i^h \bigcap N_{s_2}| \geq 1$

\end{lma}

\begin{proof}
Suppose there exists the leaf block $B_i^h$ such that $|B_i^h \bigcap N_{s_1}| + |B_i^h \bigcap N_{s_2}| < 1$.
Then, it is easy to see that node $u \in B^h_i \bigcap \mathcal{A}_i$ is an articulation points of graph $G[C_i \bigcup \{s_1,s_2\}]$, and therefore also of $G$. This constitutes a contradiction because $G(V,E)$ is bi-connected and there is no articulation point exists in $G$, and thus completes the proof.
\end{proof}

Take Fig. \ref{Fig:leaf-block-cut} for example. Block $B^3_i$ and $B^h_i$ are two leaf blocks of component $C_i$ in which $B^3_i$ has links to connect to both $s_1$ and $s_2$, and $B^h_i$ has links to connect to $s_1$. Note that if $B^h_i \bigcap N_{s_1} = \emptyset$, then node $u$ ($u \in B^h_i \bigcap \mathcal{A}_i$) is an articulation point of $G$. From the properties of Lemma \ref{lm:leaf-block}, we now characterize $\bf{(II)}$.

Again, without loss of generality, we also assume that
there exists $G_i(2)$ is a bi-connected subgraph of a leaf block $B_i^h$. Let $\overline{G_i(2)}$ be the rest of block $B_i^h$
without $G_i(2)$ such that $\overline{G_i(2)} = B_i^h \backslash G_i(2)$ (see Fig. \ref{Fig:leaf-block-cut-2}).
Furthermore, we have $\overline{G_i(2)}$ $\bigcap N_{s_1} \neq \emptyset$.

\begin{figure}
\center \subfigure{\includegraphics[width=8cm]{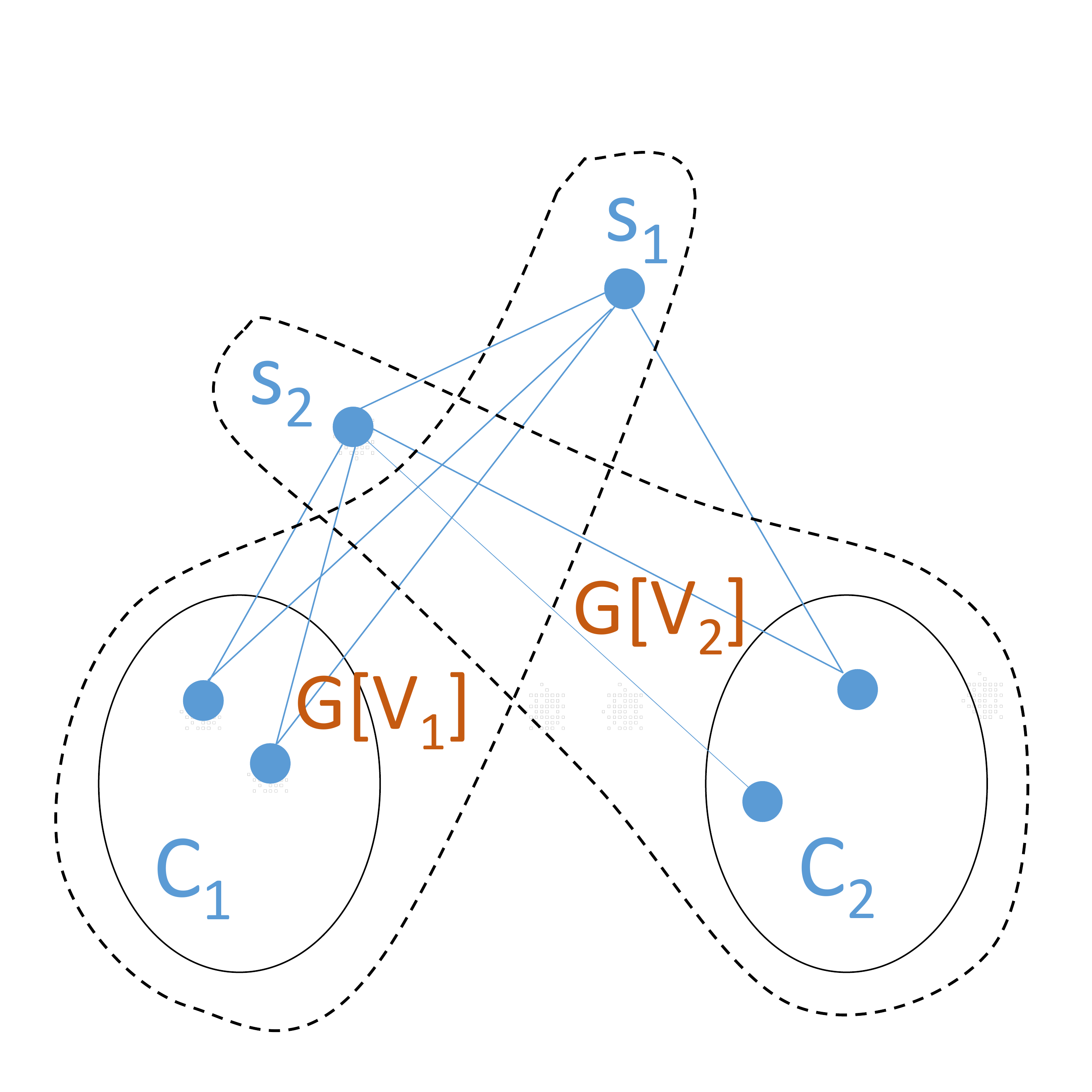}}
\caption{An example of graph $G$ with at most two $\textit{components}$.} \label{Fig:two-components}
\end{figure}

From above results, we have a complete characteristic of a bi-connected graph $G$.

\begin{cl} \label{cl:leaf-block-par}
Given a bi-connected graph $G(V,E)$ such that $\exists$ $(s_1,s_2) \in E |$ the removals of $s_1$ and $s_2$ lead to $k$ components $C_1, C_2, \ldots, C_k$, with $k > 2$, in terms of the blocks of $C_i$ ($1 \leq i \leq k$), $G$ can be bipartitioned into $2$ bi-connected subgraphs $G[V_1]$ and $G[V_2]$ with $s_1, s_2 \in V_1$, if and only if $G[V_2] = G^h_i(2)$ in which $G^h_i(2)$ is a subgraph of $B^h_i$, $B^h_i$ is a leaf block of $C_i$ ($1 \leq i \leq k$), $G^h_i(2)$ is bi-connected, and $\overline{G^h_i(2)} = B_i^h \backslash G^h_i(2)$ is bi-connected as well. We also have $\overline{G^h_i(2)} \bigcap N_{s_1} \neq \emptyset$ (otherwise node $v \in \overline{G^h_i(2)} \bigcap \mathcal{A}_i$ is an articulation of $G$). Note that $\overline{G^h_i(2)} = \emptyset$ if $G[V_2] = B^h_i$ (as shown in Fig. \ref{Fig:leaf-block-cut}).
\end{cl}

Take Fig. \ref{Fig:leaf-block-cut-2} for example. Component $C_i$ is decomposed into blocks $B_i^1, B_i^2, \ldots, B_i^h$ in which $B_i^h$ is a leaf block of $C_i$. $G^h_i(2)$ is a subgraph of $B_i^h$, $G^h_i(2)$ is bi-connected and $w(G^h_i(2)) = n_2$. $\overline{G^h_i(2)}$ is the rest of block $B_i^h$ without $G^h_i(2)$, $\overline{G^h_i(2)}$ is bi-connected, and $\overline{G^h_i(2)} \bigcap N_{s_1} \neq \emptyset$. $G$ is bipartitioned into $G[V_1]$ and $G[V_2]$ with $G[V_2] = G^h_i(2)$, and $s_1, s_2 \in V_1$.

Now we consider the case in which given a bi-connected graph $G(V,E)$, such that $\forall (s_1,s_2) \in E$, the removals of $s_1$ and $s_2$ lead to at most two components $C_1$ and $C_2$.
It is easy to see $C_k \bigcap N_{s_1} \neq \emptyset$ and $C_k \bigcap N_{s_2} \neq \emptyset$ ($k= 1,2$) (otherwise either $s_1$ or $s_2$ is an articulation point of $G$).

\begin{po} \label{po:degree3}
Suppose $G$ can be bipartitioned into 2 bi-connected subgraphs, $G[V_1]$ and $G[V_2]$, such that
$s_1 \in V_1$, $s_2 \in V_2$, and $w(V_1) = n_1$, $w(V_2) = n_2$. Without loss of generality, assume that $V_1 = C_1 \bigcup \{s_1\}$ and $V_2 = C_2 \bigcup \{s_2\}$. Then, we have $C_1 \bigcap N_{s_1} \geq 2$, and $C_2 \bigcap N_{s_2} \geq 2$.
\end{po}

From the property \ref{po:degree3} we have:

\begin{po} \label{po:degree32}
Given a bi-connected graph $G(V,E)$, suppose $G$ can be bipartitioned into 2 bi-connected subgraphs, $G[V_1]$ and $G[V_2]$, such that $s_1 \in V_1$, $s_2 \in V_2$, and $w(V_1) = n_1$, $w(V_2) = n_2$. Then, $|N_{s_1}| \geq 3, |N_{s_2}| \geq 3$.
\end{po}

\begin{cl} \label{cl:two-components}
Given a bi-connected graph $G(V,E)$ such that $\exists$ $(s_1,s_2) \in E |$ the removals of $s_1$ and $s_2$ lead to at most $2$ components $C_1, C_2$. Graph $G$ can be bipartitioned into $2$ bi-connected subgraphs $G[V_1]$ and $G[V_2]$ with $s_1 \in V_1$, $s_2 \in V_2$ if and only if $C_i \bigcap N_{s_i} \geq 2$ ($i = 1,2$), $C_i \bigcap N_{s_j} \geq 2$ ($(i,j) = \{(1, 2), (2, 1)\}$), $|C_1 \bigcup \{s_1\}| = n_1$, $|C_2 \bigcup \{s_2\}| = n_2$, and $G[V_1]$ and $G[V_2]$ are bi-connected (see Fig. \ref{Fig:two-components} for an illustration).
\end{cl}

\begin{figure}
\center
\subfigure[]{\includegraphics[width=7.5cm]{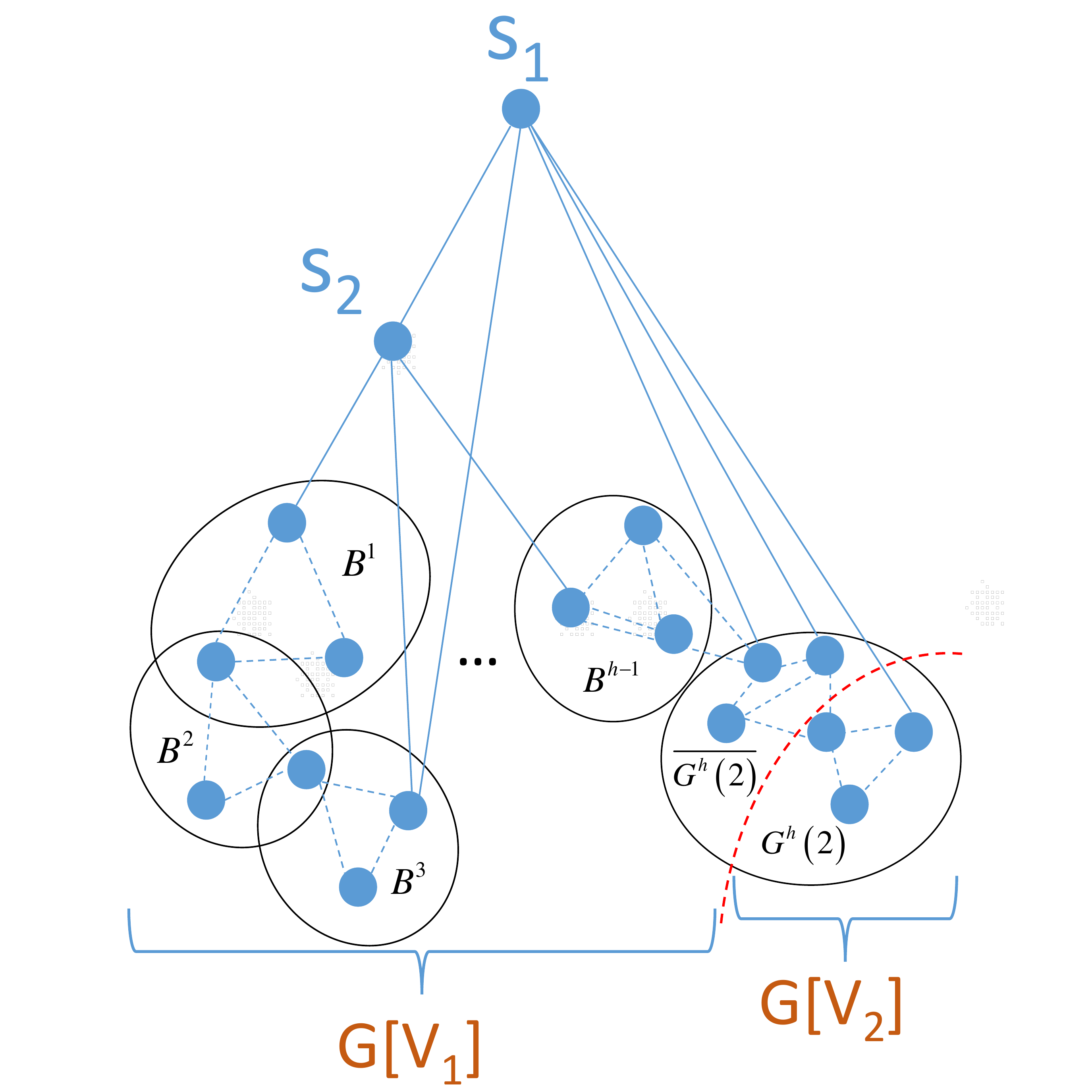} \label{Fig:one-component-together}}
\subfigure[]{\includegraphics[width=7.5cm]{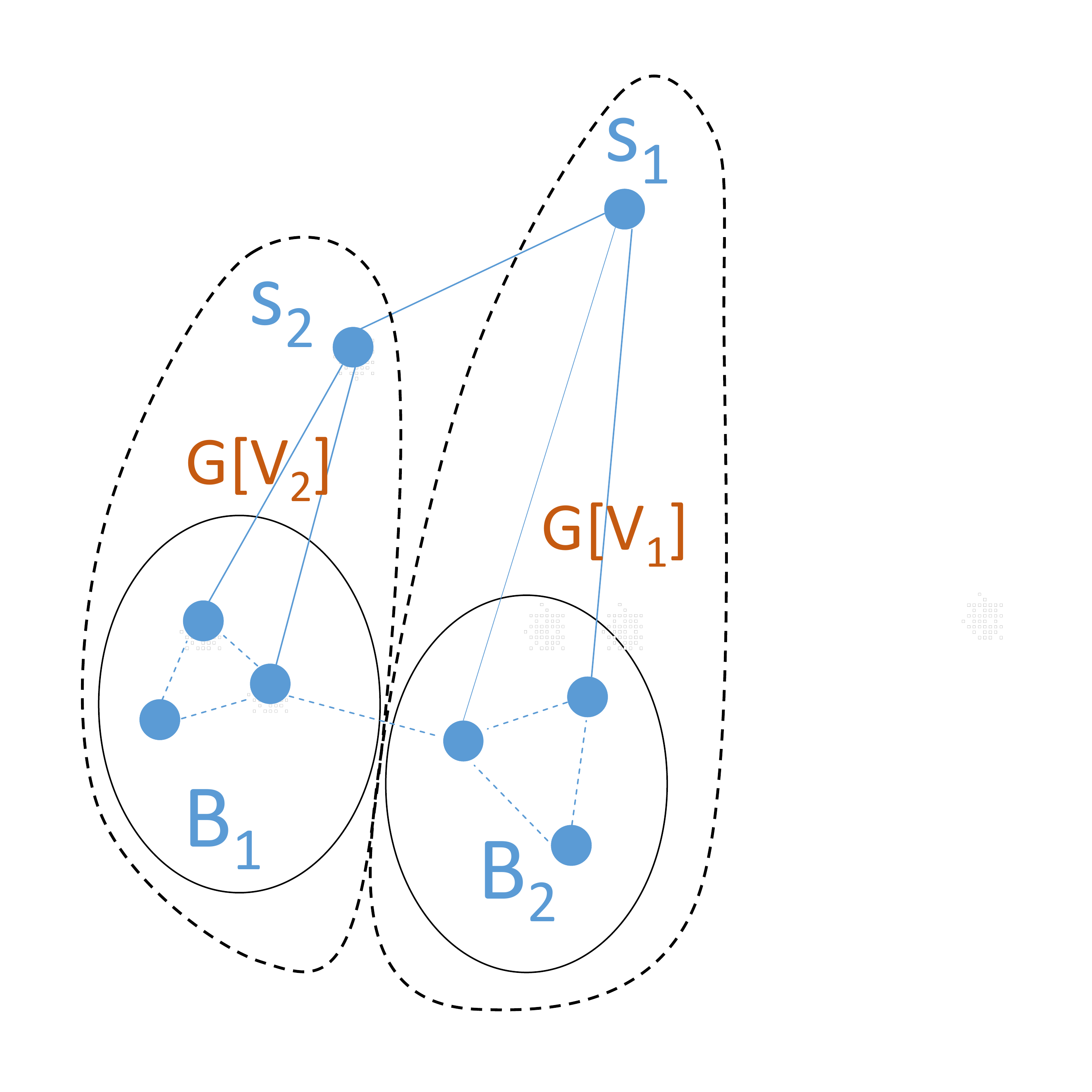} \label{Fig:one-component-separate}}
\caption{An example of graph $G$ with only one $\textit{component}$ and $s_1, s_2 \in V_1$ (a), $s_1 \in V_1, s_2 \in V_2$ (b).}
\end{figure}

We continue expanding the {\it R-BBG$_2$} problem on bi-connected graph in the third case such that
$\forall (s_1, s_2) \in E$, the removals of $s_1$ and $s_2$ lead to a single connected component $C$.
Suppose $C$ can be decomposed into $h$ blocks ($C$ is not bi-connected). Let $\mathcal{B} = \{B^1, B^2,\ldots, B^h \}$ ($h \geq 2$) be the set of blocks of the component $C$. Then, suppose $G$ can be partitioned into 2 bi-connected graphs $G[V_1]$ and $G[V_2]$ such that:

\begin{enumerate}
  \item $s_1, s_2 \in V_1$ (we also may have $s_1, s_2 \in V_2$, without loss of generality, we assume $s_1, s_2 \in V_1$). Then there exists a leaf block $B^h \in \mathcal{B}$ in which $B^h$ can be partitioned into 2 subgraphs $G^h(2)$ and $\overline{G^h(2)}$, in which $G[V_2] = G^h(2)$, $G^h(2)$ is bi-connected subgraphs, $\overline{G^h(2)} \bigcap N_{s_1} \neq \emptyset$, the graph which is induced by $\overline{G^h(2)} \bigcup \{s_1\}$ is also bi-connected (an illustration is shown in Fig. \ref{Fig:one-component-together}). Note that if $G[V_2] = B^h$ then $\overline{G^h(2)} = \emptyset$.\\

  \item $s_1 \in V_1$, and $s_2 \in V_2$ (an example is shown in Fig. \ref{Fig:one-component-separate}). Then it is easy to verify that $h = 2$ (otherwise either $s_1$ or $s_2$ will be an articulation point of $G[V_1]$ or $G[V_2]$). Since $h =2$, without loss of generality, suppose $B_2$ is a leaf block of $C$. Then, we have:
	
	\begin{itemize}
  		\item $B_2 \bigcap N_{s_1} \geq 2$ and $G[V_1]$ is induced by $B_2 \bigcup \{s_1\}$.\\
  		\item $B_1 \bigcap N_{s_2} \geq 2$ and $G[V_2]$ is induced by $B_1 \bigcup \{s_2\}$.
	\end{itemize}
\end{enumerate}

Now we consider a special case in which $C$ is bi-connected subgraph. Suppose $C$ can be decomposed into 2 bi-connected subgraphs $G_1(1)$ and $G_1(2)$. Then, $G$ can be partitioned into 2 bi-connected subgraphs $G[V_1]$ and $G[V_2]$ if either:

\begin{enumerate}
  \item $G_1(1) \bigcap N_{s_2} \geq 2$ and $G[V_1]$ is induced by $G_1(1) \bigcup \{s_2\}$, and $G_1(2) \bigcap N_{s_1} \geq 2$ and $G[V_2]$ is induced by $G_1(2) \bigcup \{s_1\}$.\\
  \item or $s_1, s_2 \in V_1$. Then $G[V_1]$ is induced by $G_1(1) \bigcup \{s_1, s_2\}$ and $G[V_2] = G_1(2)$ in which both are bi-connected.
\end{enumerate}

\begin{figure}
\center
\subfigure[]{\includegraphics[width=7cm]{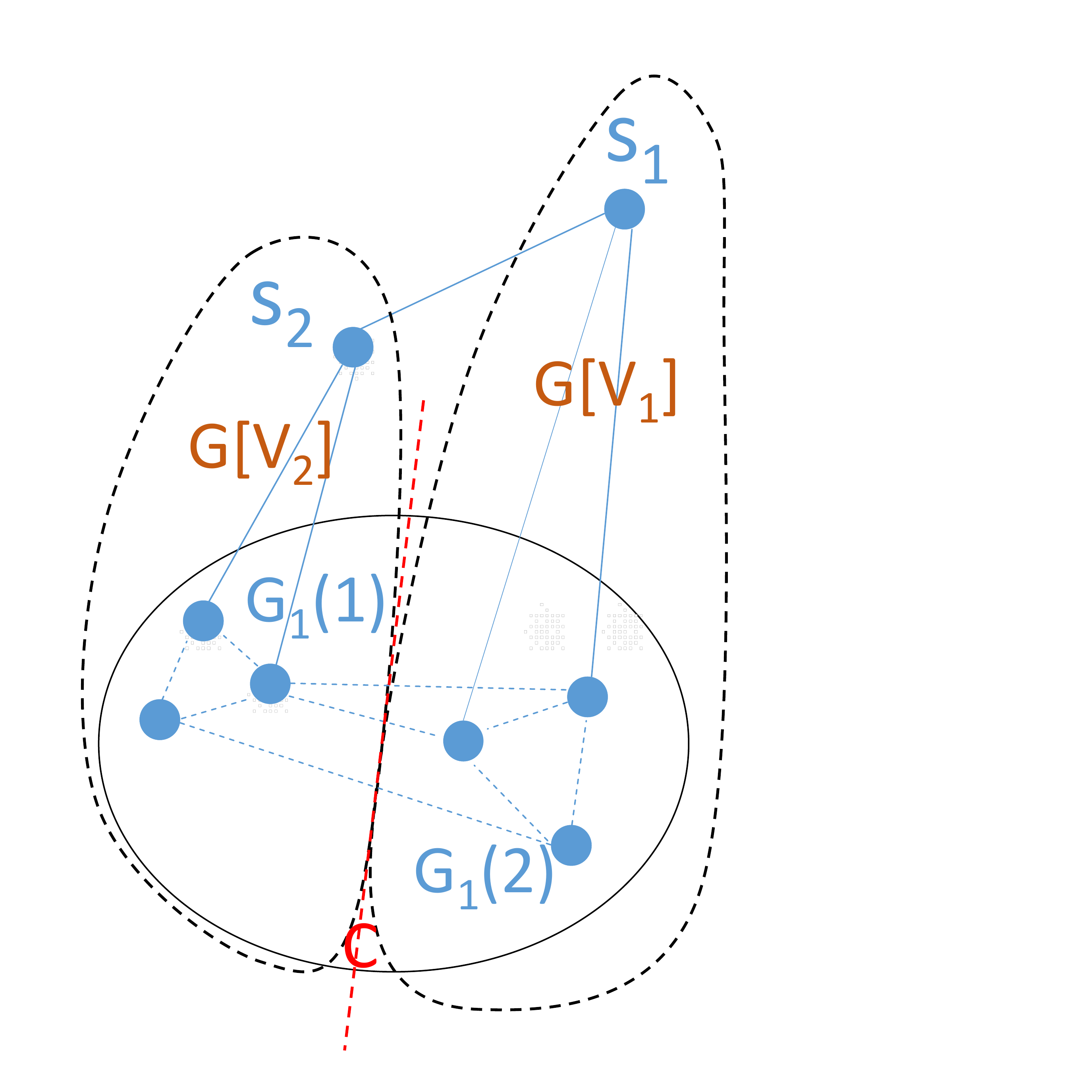} \label{Fig:one-component-one-block-a}}
\subfigure[]{\includegraphics[width=7cm]{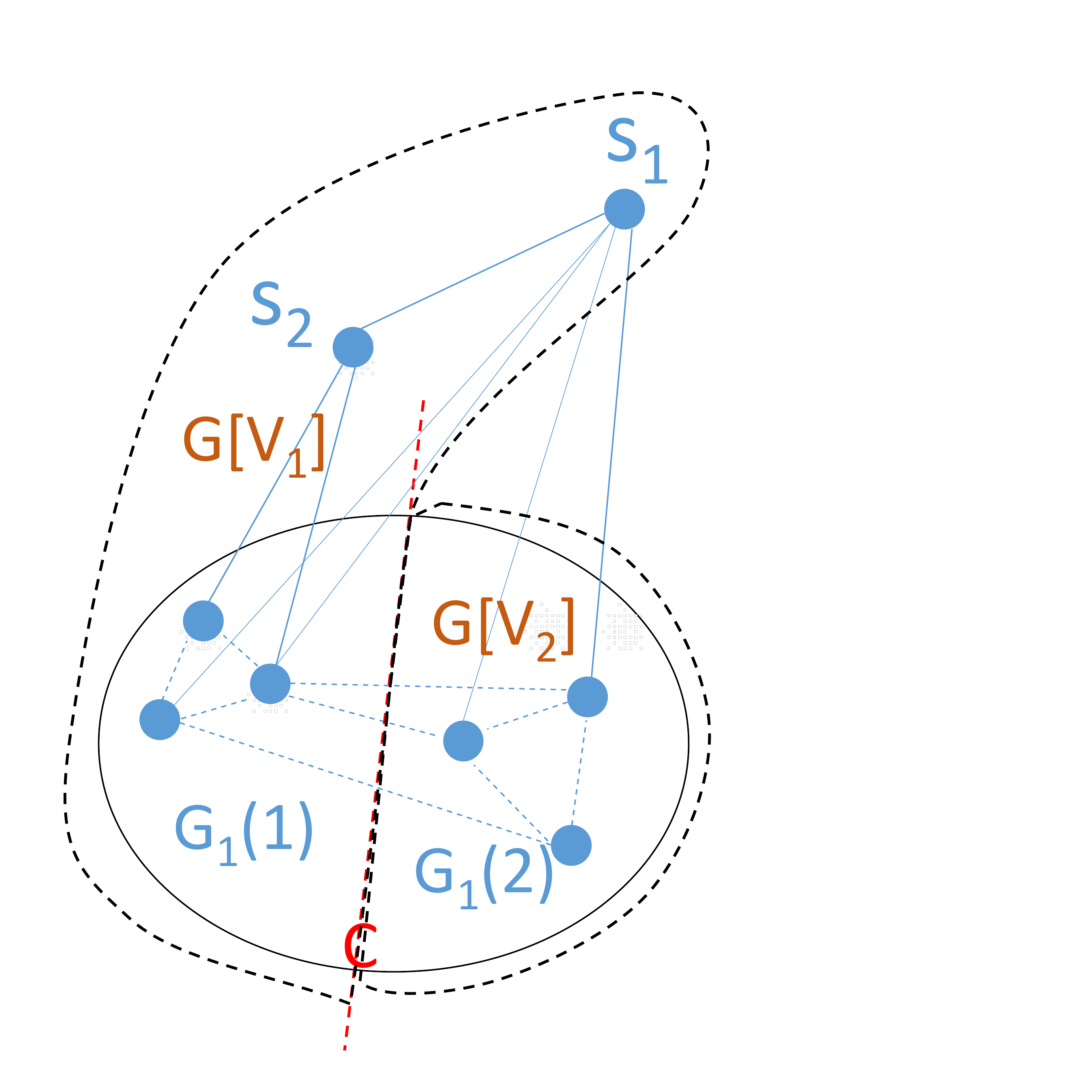} \label{Fig:one-component-one-block-b}}
\caption{An example of graph $G$ with one $\textit{component}$ $C$, and $C$ contains only one $\textit{block}$ in which $s_1 \in V_1$, $s_2 \in V_2$ (a), and $s_1, s_2 \in V_1$ (b).}
\label{Fig:one-component-one-block}
\end{figure}

For example, in Fig. \ref{Fig:one-component-one-block} the graph contains only one component $C$ if removing $s_1$ and $s_2$ and all relative links. The component $C$ then is decomposed into $2$ bi-connected subgraphs $G_1(1)$ and $G_1(2)$ in which in Fig. \ref{Fig:one-component-one-block-a} $s_1 \in V_1$, $s_2 \in V_2$, and in Fig. \ref{Fig:one-component-one-block-b} $s_1, s_2 \in V_1$.

From above results we have:

\begin{cl} \label{cl:one-two-chidren}
Given a bi-connected graph $G(V,E)$, a sequence pair of numbers $(n^1_1, n^1_2), \ldots, (n^p_1, n^p_2)$ $\in N^*$, and $(s^1_1, s^1_2), (s^2_1, s^2_2), \ldots, (s^p_1, s^p_2) \in E$ (even though we have supposed that $n^i_1 \approx n^i_2 \forall 1 \leq i \leq p$, in general we may have $n^i_1 \neq n^i_2$).

\begin{enumerate}
  \item There is feasible solution for {\it R-BBG$_2$} problem such that in every step $s^i_1 \in V^i_1$, $s^i_2 \in V^i_2$, $w(V^i_1) = n^i_1$, $w(V^i_2) = n^i_2$, and $G_i[V_1], G_i[V_2]$ are bi-connected ($G_i[V_1]$ and $G_i[V_2]$ are induced by $V^i_1$ and $V^i_2$, respectively) if either:
  		\begin{itemize}
  			\item the removals of $s^i_1$ and $s^i_2$ lead to only one component.
  			\item or the removals of $s^i_1$ and $s^i_2$ lead to at most two components.
		\end{itemize}
\end{enumerate}
\end{cl}

\section{Decision Algorithm} \label{section:decision}

\begin{figure}
\center
\subfigure{\includegraphics[width=15cm]{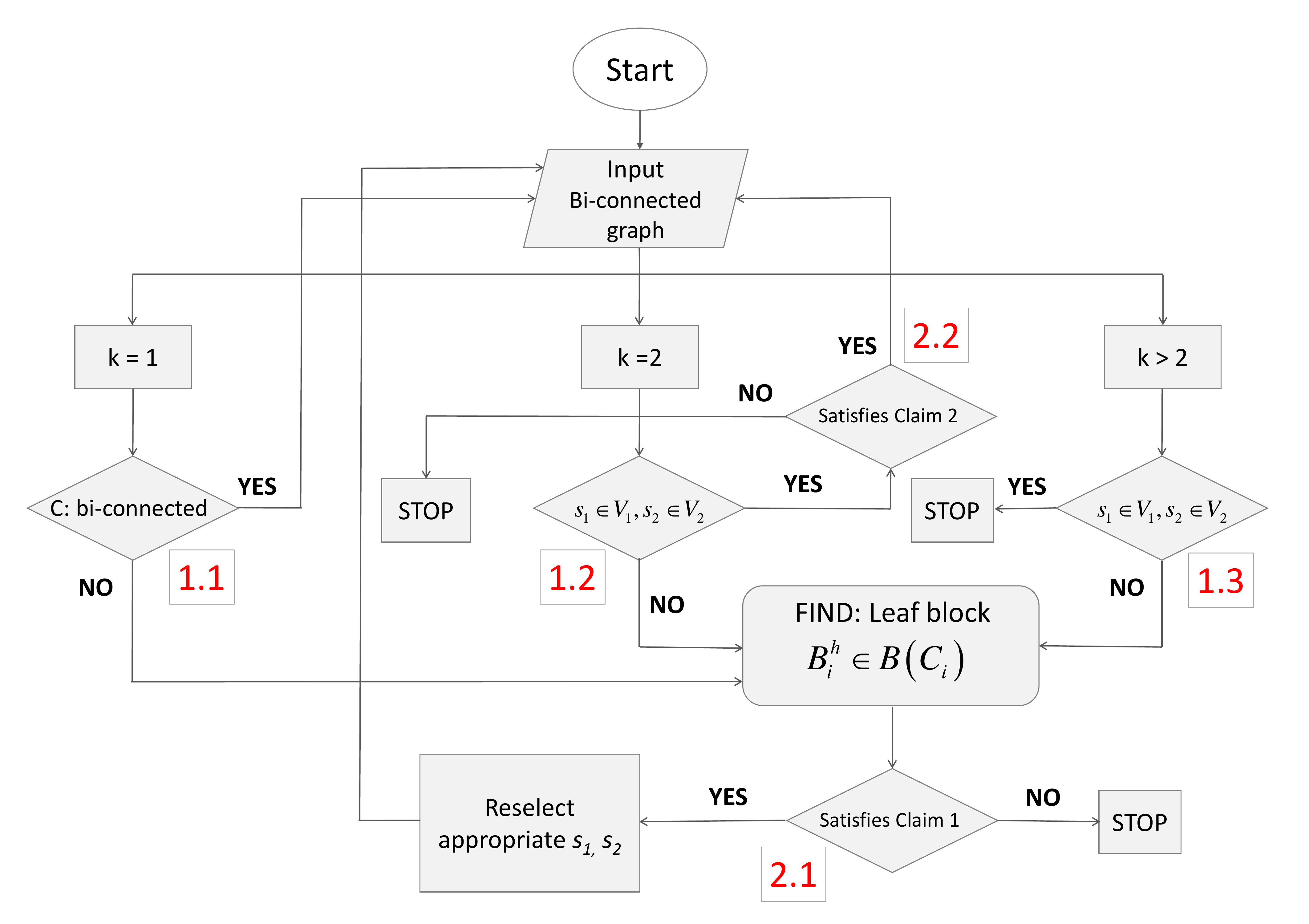}}
\caption{Flowchart of the decision algorithm.}
\label{Fig:decision}
\end{figure}

Let $C$ denote the set of components of $G$ induced by removals of $s_1$ and $s_2$.
We have $C = C_1, C_2, \ldots, C_k$. The details of the algorithm are shown in the flowchart (Fig. \ref{Fig:decision}). Referring to the flowchart, in step $\bf{1.1}$ if the unique component $C$ of graph is bi-connected, we can recursively solve the problem by solving the problem with a reducing $C$ to a graph $G'$ in which $G' = G \backslash \{s_1,s_2\}$ is still connected.

Continue examining the flowchart we see that since removals of $s_1$ and $s_2$ lead to two or more components as shown in steps $\bf{1.2}$ and $\bf{1.3}$, the requirement of $s_1, s_2$ to be on two different parts after partition causes the algorithm to be stopped in general. For $|C| \geq 2$ the only case makes the {\it bipartition decision} valid with $s_1 \in V_1, s_2 \in V_2$ is in the step $\bf{2.2}$ if the removals of $s_1$ and $s_2$ lead to at most $2$ components, and satisfies Claim \ref{cl:two-components}. The details of the behavior of the algorithm on bi-connected graph $G$ in the steps $\bf{1.2}$ and $\bf{1.3}$ indicated in Lemma \ref{lm:k-3} and Claim \ref{cl:two-components}.

Moreover, every passage through the $\bf{NO}$ branch of step $\bf{1}$ leads to finding a leaf block $B^h_i$ in component $C_i$ ($1 \leq i \leq k$). Step $\bf{2}$ of the flowchart illustrates for the behavior of the algorithm on $G$ indicated by Claim \ref{cl:leaf-block-par}.

\begin{thm} \label{lm:time-complex}
Decision algorithm solves problem with a time complexity $O(m^2nt|\mathcal{B}|)$
\end{thm}

\begin{proof}
If the algorithm stops at steps $\bf{1.3}$, $\bf{2.1}$, and $\bf{2.2}$, by Lemma \ref{lm:k-3}, Claim \ref{cl:two-components} and \ref{cl:leaf-block-par}, we know that the investigated problem is infeasible. We consider the case that problem is feasible.
In each iteration of the algorithm, it takes $O(mn)$ time to decompose every component to blocks and find leaf block. Let $|\mathcal{B}|$ be the number of leaf blocks in each iteration. In addition, let $O(t)$ be the time complexity of finding the subgraph $G^h_i(2)$ of any leaf block $B^h_i$ of any component $C_i$ such that $G^h_i(2) = G[V_2]$. We know that the time complexity of finding $G^h_i(2)$ may take exponential time; however, the number of nodes in each leaf block is very small. Because at most $m$ ($m = |E|$) iterations are in the decision algorithm, the time complexity is bounded in $O(m^2nt|\mathcal{B}|)$. This completes the proof.
\end{proof}


\section{Conclusion} \label{section:conclude}

In this paper we investigated the problem of recursive bipartition ({\it R-BBG$_2$}) of bi-connected graph. Since finding the exact solution for q-partition of $G(V,E)$ (partition $G$ into $q$ pieces) in general is a NP-hard problem, we proposed a decision algorithm to recursively bipartition graph with fairness assumption that the input graph is bi-connected. To qualify the input graph, we consider all possible cases in which graph contains one, two, and multiple components (the subgraphs were induced by the removals of two particular nodes $s_1, s_2 \in V$). Theoretical analysis of the decision algorithm shows that it can efficiently solve the {\it R-BBG$_2$} problem on bi-connected graph.

\bibliographystyle{ieeetr}
\bibliography{my}

\end{document}